\documentclass[11pt,twoside,a4paper]{article}
\usepackage[affil-it]{authblk}
\usepackage{tikz}
\usetikzlibrary{matrix}
\usepackage[cp1250]{inputenc}
\usepackage[english]{babel}
\usepackage[toc,page]{appendix}
\usepackage{tikz-cd}
\usepackage{setspace}
\usepackage{authblk}
\usepackage{amssymb,amsmath,amsfonts,amsthm,lscape,xcolor,color,enumerate}
\usepackage[color,all]{xy}

\input xy
\xyoption{all}
\xyoption{frame}

\newdir{ (}{{}*!/-5pt/@^{(}}
\newdir{|>}{!/4,5pt/@{|}*:(1,-.2)@^{>}*:(1,+.2)@_{>}}

\def\d{\mathrm d}

\def\J2{\mathsf J^2}

\newtheorem{theorem}{Theorem}

\newtheorem{lemma}[theorem]{Lemma}

\begin{document}

\title{Hamilton-Jacobi Formalism on Locally Conformally Symplectic Manifolds}

\author{O\u{g}ul Esen$^{\dagger}$, Manuel de Le\'on$^{\ddagger}$, Cristina
Sard\'on$^{*}$,  Marcin Zajac$^{**}$}
\date{}
\maketitle

\centerline{Department of Mathematics$^{\dagger}$,}
\centerline{Gebze Technical University, 41400 Gebze, Kocaeli, Turkey.}
\centerline{oesen@gtu.edu.tr}
\vskip 0.5cm

\centerline{Instituto de Ciencias Matem\'aticas, Campus Cantoblanco$^{\ddagger}$,}
\centerline{Consejo Superior de Investigaciones Cient\'ificas}
\centerline{and}
\centerline{
Real Academia Espa{\~n}ola de las Ciencias.}
\centerline{C/ Nicol\'as Cabrera, 13--15, 28049, Madrid, Spain.}
\centerline{mdeleon@icmat.es}
\vskip 0.5cm

\centerline{Instituto de Ciencias Matem\'aticas, Campus Cantoblanco$^{*}$,}
\centerline{Consejo Superior deInvestigaciones Cient\'ificas.}
\centerline{C/ Nicol\'as Cabrera, 13--15, 28049, Madrid, Spain.}
\centerline{cristinasardon@icmat.es}
\vskip 0.5cm

\centerline{Department of Mathematical Methods in Physics$^{**}$,}
\centerline{Faculty of Physics. University of Warsaw,}
\centerline{ul. Pasteura 5, 02-093 Warsaw, Poland.}
\centerline{
marcin.zajac@fuw.edu.pl}

\begin{abstract}

In this article we provide a Hamilton-Jacobi formalism in locally conformally symplectic manifolds. Our interest in the Hamilton-Jacobi theory comes from the suitability of this theory as an integration method for dynamical systems, whilst our interest in the locally conformal character will account for physical theories described by Hamiltonians defined on well-behaved line bundles, whose dynamic takes place in open subsets of the general manifold. We present a local l.c.s. Hamilton-Jacobi in subsets of the general manifold, and then provide a global view by using the Lichnerowicz-deRham differential. We show a comparison between the global and local description of a l.c.s. Hamilton--Jacobi theory, and how actually the local behavior can be glued to retrieve the global behavior of the Hamilton-Jacobi theory.
\end{abstract}

\newpage 
\tableofcontents

\section{Introduction}

We aim at pursuing physical theories whose Hamiltonians or dynamical variables are
defined in open subsets of a manifold. In those subsets, these systems behave
like symplectic mechanical systems, although the complete global dynamics fails to be symplectic. This phenomenom appears, for example, in some physical theories with nonlocal potentials, or mechanical systems that are defined by parts, each of which behaves differently and accordingly to different laws. 
Systems with such local and global characteristics will be referred to as
dynamical systems on a locally conformally symplectic (l.c.s.) manifolds, from a geometric point of view. Let us explain this local setting more explicitly.

As we know, symplectic geometry is somehow a global thing \cite{AbMa78,Arnold,LeRo89,Salamon}. There exists a Darboux theorem asserting that, locally, two symplectic manifolds can not be distinguished
from one another. Nonetheless, only one of the two conditions
ensuring that a two-form on an even-dimensional manifold is symplectic, is a global property, that is, closure.  This property imposes strong cohomological restrictions on the existence of a symplectic structure on an even-dimensional compact manifold, as it is that all the Betti numbers of even-degree must be non-zero. This is why determining which compact manifolds admit a symplectic structure is still an open problem.
Hence, on a more local note, Hwa-Chung Lee in 1941 \cite{Hwa} reconsidered the general setting of an even-dimensional endowed with a non-degenerate two-form $\omega$. First, he discussed the closed case, i.e., the symplectic case, and then the problem of two two-forms $\omega_1$ and $\omega_2$, which are conformal to one another. 
Later in 1985, Vaisman \cite{vaisman} defined a locally conformally symplectic (l.c.s.) manifold
as an even dimensional manifold endowed with a non-degenerate two-form such that for every point $p\in M$ there is an
open neighborhood $U$ such that
\begin{equation} \label{cond}
d\left(e^{-\sigma}\omega|_{U}\right)=0
\end{equation}
where $\sigma:U\rightarrow \mathbb{R}$ is a smooth function. If this condition
holds for $U=M$, then $(M,\omega)$ turns out to be a globally
conformally symplectic manifold. If \eqref{cond} holds for a constant function $\sigma$, then $(M,\omega)$ becomes a symplectic
manifold. The work of Lee \cite{Hwa} proposes an equivalent definition with the aid of
a compatible one-form, named the Lee form. We will introduce this definition in subsequent sections. 
It is important to notice that at a local scale, a symplectic manifold can not be
distinguished from a l.c.s. manifold. Thus, not only all symplectic
manifolds locally look alike, but potentially there may exist
manifolds which locally look like symplectic manifolds and however fail to do so globally \cite{Banzoni}.

In this paper we propose the resolution of globally nonsymplectic Hamiltonian systems that are defined locally by a symplectic structure by adding a conformal factor to the general Hamiltonian. The Hamiltonian itself without being restricted to a conformal factor, is not a Hamiltonian in the proper sense of symplectic geometry, i.e., there does not exist a symplectic form that is nondegenerate and closed. Instead, there is a nondenegerate two-form and a closed one-form that define a locally conformally symplectic structure. In the following sections, we will describe how this particulary works. Hamilton equations in the locally conformal symplectic scenario were studied by I. Vaisman \cite{vaisman}, and later extended by Chinea, de Le\'on and Marrero \cite{ChLeMa91} for the time-dependent case (cosymplectic framework).
It is our intention to study the dynamics from the local symplectic point of view, and from a global l.c.s. view. The so called Lichnerowicz-deRham differential will be a geometric object that will help us compare the local and global setting \cite{LMP}.

 On the other hand, to study the dynamics, we aim at constructing
a Hamilton-Jacobi theory. Why are we interested in a Hamilton-Jacobi theory? We must say that it constitutes the third
most important theory in classical mechanics, after the Lagrangian or the Newtonian picture. Nonetheless, it is more unfrequent, but the solvability of the Hamilton-Jacobi equation under certain circumstances is very convenient, e.g., in case of separable type potentials. Let us review this briefly. The well-known time-independent Hamilton-Jacobi equation (HJE)
\begin{equation}\label{HJeq1}
H\left(q^i,\frac{\partial W}{\partial q^i}\right)=E
\end{equation}
\noindent
was interpreted geometrically with a primordial observation on a symplectic phase space: if a Hamiltonian
vector field $X_{H}:T^{*}Q\rightarrow TT^{*}Q$ is projected into a vector field $X_H^{dW}:Q\rightarrow TQ$ on a lower
dimensional manifold by means of a one-form $dW$, then the integral curves of the projected
vector field $X_{H}^{dW}$ can be transformed into integral curves of $X_{H}$ provided that $W$ is a solution of \eqref{HJeq1}. If we define the projected vector field as:
\begin{equation}
 X_H^{dW}=T\pi_Q\circ X_H\circ dW,
\end{equation}
where $T{\pi}_Q$ is the induced projection on the tangent space, $T{\pi}_Q:TT^{*}Q\rightarrow T^{*}Q$ by the canonical projection
$\pi_Q:T^{*}Q\rightarrow Q$, it implies the commutativity of the diagram below:

\[
\xymatrix{ T^{*}Q
\ar[dd]^{\pi_Q} \ar[rrr]^{X_H}&   & &TT^{*}Q\ar[dd]^{T\pi_Q}\\
  &  & &\\
 Q\ar@/^2pc/[uu]^{dW}\ar[rrr]^{X_H^{dW}}&  & & TQ}
\]
\noindent
Notice that the image of $dW$ is a Lagrangian submanifold, since it is exact and consequently, closed. Indeed, one can change $dW$ by a closed one-form $\gamma$, since it is equivalent to the image be a Lagrangian submanifold.
Lagrangian submanifolds are very important objects in Hamiltonian mechanics, since the dynamical equation (Hamiltonian or Lagrangian)
can be described as Lagrangian submanifolds of convenient symplectic manifolds.
We enunciate the following theorem \cite{CaGrMaMaMuRo06}.
\begin{theorem}
\label{HJT-pre} For a closed one-form $\gamma=dW$ on $Q$ the following conditions are
equivalent:
\begin{enumerate}
\item The vector fields $X_{H}$ and $X_{H}^{\gamma }$ are $\gamma$-related,
that is
\begin{equation}
T\gamma \circ X_H^{\gamma}=X_H\circ\gamma.
\end{equation}
\item The following equation is fulfilled 
\[
d\left( H\circ \gamma \right)=0.
\]
\end{enumerate}
\end{theorem}
\noindent
The first item in the theorem says that if $\left( q^i\left(
t\right) \right) $ is an integral curve of $X_{H}^{\gamma }$, then $\left(
q^i\left( t\right) ,\gamma_j \left( q\left( t\right) \right) \right) $ is an
integral curve of the Hamiltonian vector field $X_{H}$, hence a solution of
the Hamilton equations \eqref{HJeq1}. Here it is assumed that $\gamma=\gamma_i(q)dq^i$. 
In other words, the solution in the complete manifold can be obtained from a solution
defined on a proper submanifold of the initial. The solution in the former initial space is retrieved by composing with a section of the bundle defined between the initial manifold, the submanifold and the projection. In the local picture, the second condition implies that exterior
derivative of the Hamiltonian function on the image of $\gamma $ is closed,
that is, $H\circ \gamma $ is constant. Under the assumption that $\gamma$ is closed, we can find (at least
locally) a function $W$ on $Q$ satisfying $dW=\gamma $. In this case, one arrives at the classical realization of the theory in (\ref{HJeq1}).

This is exactly the procedure that we will use to obtain a Hamilton-Jacobi equation on a l.c.s. manifold. For it, we will need to redefine the concept of Lagrangian submanifold in l.c.s. manifolds. Also, the dynamics, as mentioned, will take place in two different scenarios: a local picture described in terms of local subsets of $T^{*}Q$, and a global picture with the aid of the Lichnerowicz-deRham differential. So, here is the planning of the paper. In Section 2, we review the geometric fundamentals of l.c.s. manifolds. We recall the musical mappings, the Lichnerowicz-deRham differential, and construct l.c.s. structures on a cotangent bundle. We introduce the concept of Lagrangian submanifolds on l.c.s. structures to start describing the dynamics on l.c.s. cotangent bundles.
Once the dynamic is described, we propose a Hamilton--Jacobi theory for l.c.s. manifolds on Section 3: first locally, and then globally. Further, we make a comparison between the local and the global Hamilton--Jacobi theories, and see how one retrieves the other by local and global considerations. To shed some light on the theory, Section 4 describes a physical example for which the dynamics is retrieved using both the local and the global setting of l.c.s. manifolds. To close the paper, Section 5 describes the possibility of finding complete solutions to the Hamilton--Jacobi equation on a l.c.s. manifold.

\section{Geometry of locally conformally symplectic manifolds}

\subsection{Basics of locally conformally symplectic manifolds} \label{blcs}

The pair $(M,\omega)$ where $\omega$ is a non-degenerate two-form is called an almost symplectic manifold, and $\omega$ is an almost symplectic two-form. If, $\omega$ is additionally closed, then the manifold turns out to be a symplectic manifold. There is an intermediate step between symplectic manifolds and almost symplectic manifolds: these are the so called locally conformally symplectic (l.c.s.) manifolds \cite{vaisman}. An almost symplectic manifold $(M,\omega)$ is a l.c.s. manifold if the two-form is closed locally up to a conformal parameter, i.e., if there exists an open neighborhood, say $U_\alpha$, around each point $x$ in $M$, and a function $\sigma_{\alpha}$ such that the exterior derivative $d(e^{-\sigma_{\alpha}}\omega\vert_{\alpha})$ vanishes identically on $U_\alpha$. Here, $\omega\vert_{\alpha}$ denotes the restriction of the almost symplectic structure $\omega$ to the open set $U_\alpha$. The positive character of the exponential function implies that the local two-form $e^{-\sigma_{\alpha}}\omega\vert_{\alpha}$ is non-degenerate as well. Being a closed and non-degenerate,
\begin{equation}\label{ua}
    \omega_{\alpha}=e^{-\sigma_{\alpha}}\omega\vert_{\alpha}
\end{equation}
is a symplectic two-form. That is, the pair $(U_\alpha, \omega_{\alpha})$ is a symplectic manifold. 

The question now is how to glue the behavior in all local open charts to arrive at a global definition for l.c.s. manifolds. Notice that, $\omega\vert_{\alpha}$ is a local realization of the global two-form $\omega$ whereas, up to now, $\omega_{\alpha}$ is defined only on $U_\alpha$. In another local chart, say $U_{\beta}$, a local symplectic two-form is defined to be $\omega_{\beta}:=e^{-\sigma_{\beta}}\omega\vert_{\beta}$. This gives that, for overlapping charts, the local symplectic two-forms are related by $\omega_{\beta}=e^{-(\sigma_{\beta}-\sigma_{\alpha})}\omega_{\alpha}$. Accordingly, this conformal relation determines scalars
\begin{equation} \label{transition}
\lambda_{\beta \alpha}=e^{\sigma_{\alpha}}/e^{\sigma_{\beta}}=e^{-(\sigma_{\beta}-\sigma_{\alpha})}
\end{equation}
 satisfying the cocycle condition 
\begin{equation}\label{cocycle}
    \lambda_{\beta \alpha}\lambda_{\alpha \gamma}=\lambda_{\beta \gamma}.
\end{equation}
This way one can glue the local symplectic two-forms $\omega_\alpha$ to a line bundle $L \mapsto M$ valued two-form $\tilde{\omega}$ on $M$. To sum up, we say that there are two global two-forms $\omega$ (real valued) and $\tilde{\omega}$ (line bundle valued) on $M$ with local realizations $\omega\vert_\alpha$ and $\omega_\alpha$, respectively. These local two-forms are related as in (\ref{ua}).    

Now, recalling \eqref{ua} once more, it is easy to see that $d\omega\vert_{\alpha}=d\sigma_{\alpha}\wedge \omega\vert_{\alpha}$, but equally,
$d\omega\vert_{\beta}=d\sigma_{\beta}\wedge \omega\vert_{\beta}$ on an overlapping region $U_\alpha \cap U_\beta$. This implies that $$d(\sigma_{\beta}-\sigma_{\alpha})\wedge \omega|_{U_{\alpha}\cap U_{\beta}}=0,$$ and since $\omega$ is nondegenerate, necessarily, $d\sigma_{\alpha}=d\sigma_{\beta}$. So that, $\theta=d\sigma_{\alpha}$ is a well defined one-form on $M$ that satisfies $d\omega=\theta\wedge \omega$. Such a one-form $\theta$ is called the Lee one-form \cite{Hwa}. Since $\theta$ is locally exact, then it is closed.  A l.c.s. manifold $(M,\omega,\theta)$ is 
a globally conformally symplectic (g.c.s.) manifold if the Lee form $\theta$ is an exact one-form. Since it is fulfilled that in two-dimensional manifolds every closed form is exact, two-dimensional l.c.s. manifolds are g.c.s. manifolds. 
Notice that the Lee form $\theta$ is completely determined by $\omega$ for manifolds with dimension $4$ and more. We can likewise denote a l.c.s. manifold by a triple $(M,\omega,\theta)$. Equivalently, this realization of locally conformally symplectic manifolds reads that a l.c.s. manifold is a symplectic manifold if and only if the Lee form $\theta$ vanishes identically. 
 Conversely, if $(M, \omega, \theta)$ is a triple such that $\omega$ is an almost symplectic form, and $\theta$ is a closed one-form such that $d\omega = \theta \wedge \omega$, then one can find an open cover $\{U_\alpha\}$ of $M$ such that, on each chart $U_\alpha$, $\theta = d \sigma_\alpha$ for some functions $\sigma_\alpha$. It is clear now that $e^{-\sigma_\alpha} \omega\vert_\alpha$ is symplectic on $U_\alpha$.\bigskip

\noindent 

\noindent \textbf{Musical mappings.} Consider an almost symplectic manifold $(M,\omega)$. The non-degeneracy of the two-form $\omega$ leads us to define a musical isomorphism 
\begin{equation} \label{mus-iso}
\omega^\flat:\mathfrak{X}(M)\longrightarrow \Omega^1(M): X \mapsto \iota_X\omega,
\end{equation}
where $\iota_X$ is the interior derivative. Here, $\mathfrak{X}(M)$ is the space of vector fields on $M$ whereas $\Omega^1(M)$ is the space of one-form sections on $M$. 
We denote the inverse of the isomorphism (\ref{mus-iso}) by $\omega^\sharp$. When pointwise evaluated, the musical mappings $\omega^\flat$ and $\omega^\sharp$ induce isomorphisms from $TM$ to $T^*M$, and from $T^*M$ to $TM$, respectively. We shall use the same notation for the induced isomorphisms. 

Let us now concentrate on the particular case of the l.c.s. manifolds.  Assume a l.c.s. manifold $(M, \omega)$ with a Lee form $\theta$. 
Referring to the $\omega^\sharp$, we define a (so called Lee) vector field 
\begin{equation} \label{Lee-v-f}
Z_\theta:=\omega^\sharp(\theta), \qquad \iota_{Z_\theta}\omega=\theta
\end{equation}
where $\theta$ is the Lee-form. By applying $\iota_{Z_\theta}$ to the both hand side of the second equation, one obtains that $\iota_{Z_\theta}\theta=0$. Further, by a direct calculation, we see that $\mathcal{L}_{Z_\theta}\theta=0$ and that $\mathcal{L}_{Z_\theta}\omega=0$. Here, $\mathcal{L}_{Z_\theta}$ is the Lie derivative.
\bigskip
 
\noindent \textbf{The Lichnerowicz-deRham differential.}
Consider now an arbitrary manifold $M$ equipped with a closed one-form $\theta$.  The Lichnerowicz-deRham differential (LdR) on the space of differential forms $\Omega(M)$ is defined as
\begin{equation} \label{LdR-Diff}
d_\theta: \Omega^k(M) \rightarrow \Omega^{k+1}(M) : \beta \mapsto d\beta-\theta\wedge\beta,
\end{equation}
where $d$ denotes the exterior (deRham) derivative \cite{GuLi84}. Notice that $d_\theta$ is a differential operator of order $1$. That is, if $\beta$ is a $k$-form then $d_\theta\beta$ is $k+1$-form. The closure of the one-form $\theta$ reads that $d_{\theta}^2=0$. This allows the definition of  cohomology as the $d_\theta$ cohomology in $\Omega(M)$
\cite{HaRy99}. We represent this by the pair $(\Omega(M),d_{\theta})$. 
A direct computation shows that an almost symplectic manifold $(M, \omega)$ equipped with a closed one-form $\theta$ is a l.c.s. manifold if and only if $d_\theta \omega = 0$. Our point is to use the LdR to establish a HJ theory on l.c.s. cotangent bundles. 
\bigskip

\noindent \textbf{{Lagrangian Submanifolds of l.c.s. manifolds.} }
Consider an almost symplectic manifold $(M, \omega)$. Let $L$ be a submanifold of $M$.
The complement $TL^{\bot}$ is defined with respect to $\omega$. For a point $x\in L$,
\begin{equation}\label{lagsub}
    T_{x}L^{\bot}=\{u\in T_xM \enskip|\enskip \omega(u,w)=0, \forall w\in T_xL\}.
\end{equation}
\noindent 
We say that $L$ is isotropic if $TL\subset TL^{\bot}$, it is coisotropic if $TL^{\bot}\subset TL$ and it is Lagrangian if $TL^{\bot}=TL$. Accordingly a submanifold is Lagrangian if it is both isotropic and coisotropic. Observe that the definition and result are exactly the same that in the symplectic case, since they are obtained at the linear level.

\subsection{Locally conformally symplectic structures on cotangent bundles.} \label{Sec-lcs-cot}
Since we are interested in constructing a HJ theory for mechanics,  we shall depict the framework on the cotangent bundles \cite{Banyaga,ChMu17,HaRy99,OtSt15}. Start with the canonical symplectic manifold $(T^*Q,\Omega_Q)$. Here, the canonical symplectic two-form $\Omega_Q = - d \Theta_Q$ is minus of the exterior derivative of the canonical Liouville one-form $\Theta_Q$  on $T^*Q$. Let $\vartheta$  be a closed one-form on the base manifold $Q$ and pull it back to $T^*Q$ by means of the cotangent bundle projection $\pi_Q$. This gives us a closed semi-basic one-form $\theta=\pi_Q^*(\vartheta)$. By means of the Lichnerowicz-deRham differential, we define a two-form 
\begin{equation} \label{omega_theta}
\Omega_\theta=-d_\theta(\Theta_Q)= -d\Theta_Q+\theta\wedge \Theta_Q=\Omega_Q+\theta\wedge \Theta_Q
\end{equation}
on the cotangent bundle $T^*Q$. Since $d\Omega_\theta=\theta\wedge \Omega_\theta$ holds, the triple 
\begin{equation} \label{T*_Q}
T^*_\theta Q=(T^*Q,\Omega_\theta,\theta)
\end{equation} 
determines a locally conformal symplectic manifold with the Lee-form $\theta$. In short, we denote this l.c.s. manifold by simply $T^*_\theta Q$. This structure is
conformally equivalent to a symplectic manifold if and only if $\vartheta$ lies in the zeroth class of the first deRham cohomology on $Q$. Notice that $T^*_\theta Q$ is an exact locally conformal symplectic manifold since $\Omega_\theta$ is defined to be minus of the Lichnerowicz-deRham differential $d_\theta$ of the canonical one-form $\Theta_Q$. It is important to note that all l.c.s. manifolds locally look like $T^*_\theta Q$ for some $Q$ and for a closed one-form $\vartheta$  \cite{ChMu17,OtSt15}. 
\bigskip

\noindent 
Consider the l.c.s. manifold $T_\theta^*Q$ in (\ref{T*_Q}) with Lee form $\theta=\pi_Q^*\vartheta$. Let $\gamma$ be a section of the cotangent bundle or, in other words, a one-form on $Q$. A direct computation shows that the pull-back of the l.c.s. structure is $d_\theta$ exact, that is 
\begin{equation} \label{LagSubT*Q}
\gamma^* \Omega_\theta = - d_\vartheta \gamma
\end{equation} 
where $d_\vartheta$ denotes the LdR differential defined by the one-form $\vartheta$ on $Q$. This implies that the image space of $\gamma$ is a Lagrangian submanifold of $T_\theta^*Q$ if and only if $d_\vartheta \gamma=0$. Since $d_\vartheta^2$ is identically zero, the image space of the one-form $d_\vartheta f$ is a
 Lagrangian submanifold of $T^*_\theta Q$ for some function $f$ defined on $Q$. 
\subsection{Dynamics on locally conformally symplectic manifolds }

Let us now concentrate on the Hamiltonian dynamics on l.c.s. manifolds \cite{vaisman13}.  As discussed previously, there are two equivalent definitions of l.c.s. manifolds. One is local, and the other is global. First consider the local definition by recalling the local symplectic manifold $(U_\alpha,\omega_\alpha)$. For a Hamiltonian function $h_\alpha$ on this chart, we write the Hamilton equations by  
\begin{equation}\label{geohamalpha}
    \iota_{X_{\alpha}}\omega_{\alpha}=dh_{\alpha}.
\end{equation}
Here, $X_{\alpha}$ is the local Hamiltonian function associated to this framework. In terms of the Darboux' coordinates $(q^i_{(\alpha)},p_i^{(\alpha)})$ on $U_\alpha$. The local symplectic two-form is $\omega_{\alpha}=dq^i_{(\alpha)}\wedge dp_{i}^{(\alpha)}$, and the Hamilton equation (\ref{geohamalpha}) becomes  
\begin{equation} \label{HamEqLoc}
    \frac{dq^i_{\alpha}}{dt}=\frac{\partial h_{\alpha}}{\partial p_i^{\alpha}},\qquad \frac{dp_i^{\alpha}}{dt}=-\frac{\partial h_{\alpha}}{\partial q_{\alpha}^i}.
\end{equation}

We have discussed the gluing problem of the local symplectic manifolds but we have not addressed this problem for the local Hamiltonian functions. We wish to define a global realization of the local Hamiltonian functions in such a way that the structure of the local Hamilton equations (\ref{HamEqLoc}) does not change under transformations of coordinates. This can be rephrased as to establish a global realization of the local Hamiltonian function $h_\alpha$ by preserving the local Hamiltonian vector fields $X_\alpha$ in (\ref{geohamalpha}). A direct observation reads that multiplying both hand sides of (\ref{geohamalpha}) by the scalars $\lambda_{\beta\alpha}$ defined in (\ref{transition}) leaves the dynamics invariant. So that, the transition 
$h_{\beta}=e^{\sigma_{\alpha}-\sigma_{\beta}}h_{\alpha}$ is needed for the preservation of the structure of the equations. In the light of the cocyle character of the scalars shown in (\ref{cocycle}), we can glue the local Hamiltonian functions $h_\alpha$ to define a section $\tilde{h}$ of the line bundle $L\mapsto M$. On the other hand, in the light of the identity $e^{\sigma_{\alpha}}h_{\alpha}=e^{\sigma_{\beta}}h_{\beta}$ one arrives at a real valued Hamiltonian function 
\begin{equation} \label{glueHamFunc}
h\vert_\alpha=e^{\sigma_\alpha}h_\alpha.
\end{equation}
on $U_\alpha$ that defines a real valued function $h$ on the whole $M$. Recall the discussion in Subsection \ref{blcs} about the local and global character of the two-forms. Similarly, we argue that there exist two global functions $\tilde{h}$ (line bundle valued) and $h$ (real valued) on the manifold $M$. On a chart $U_\alpha$, these functions reduce to $h_\alpha$ and $h\vert_\alpha$, respectively, and they satisfy the relation \eqref{glueHamFunc}.

In order to recast the global picture of the Hamilton equation (\ref{geohamalpha}), we first substitute the identity (\ref{glueHamFunc}) into \eqref{geohamalpha}. Hence, a direct calculation turns the Hamilton's equations into the following form
\begin{equation} \label{LocHam2}
\iota_{X_\alpha} \omega\vert_\alpha =dh\vert_\alpha-h\vert_\alpha d\sigma_\alpha
\end{equation}
where we have employed the identity (\ref{ua}) on the left hand side of this equation. Notice that all the terms in equation (\ref{LocHam2}) have global realizations. So, we can write
\begin{equation}\label{semiglobal}
   \iota_{X_{h}}\omega=dh-h\theta,
\end{equation}
where $X_{h}$ is the vector field obtained by gluing all the vector fields $X_\alpha$. That is, we have $X_h\vert_\alpha=X_\alpha$. Notice that \eqref{semiglobal} can also be written as 
\begin{equation}\label{semiglobal3}
\iota_{X_h} \omega = d_\theta h,
\end{equation}
where $d_\theta$ is the Lichnerowicz-deRham differential given in (\ref{LdR-Diff}). The vector field $X_h$ defined in (\ref{semiglobal3}) is called as Hamiltonian vector field for the Hamiltonian function $h$. In terms of the Lee vector field $Z_\theta$ defined in (\ref{Lee-v-f}), the Hamiltonian vector field is computed to be 
\begin{equation}\label{vflcs}
X_h=\omega^\sharp(dh)+hZ_\theta,
\end{equation}
where $\omega^\sharp$ is the musical isomorphism induced by the almost symplectic two-form $\omega$. From this, one can easily see that, apart from the classical symplectic framework, for the constant function $h = 1$ the corresponding Hamiltonian vector field is not zero but the Lee vector in (\ref{Lee-v-f}), that is $Z_\theta=X_1$. More generally, a vector field $X$ is called a locally Hamiltonian vector field if 
\begin{equation}
d _ { \theta } ( \iota_{X}\omega)=0.
\end{equation}
It is immediate to see that a Hamiltonian vector field is locally Hamiltonian vector field since $d^2_\theta =0$.  

\subsection{Jacobi structures and l.c.s. manifolds} \label{Jacobi}

Let $M$ be a manifold, and $\{\bullet,\bullet\}$ be a bracket on the space of smooth functions on $M$. Assume that, this bracket is skew-symmetric and satisfies the Jacobi identity \cite{Marle91} and it is local, i.e., the support of the bracket of two function lies in the intersection of the supports of those functions.  Now, a manifold equipped with such bracket is a Jacobi manifold \cite{Marle91}. 

Consider a manifold $M$ with a bivector field $\Lambda$ and a vector field $E$ satisfying the conditions
\begin{equation} \label{JB-CC} 
[\Lambda,\Lambda]=2E\wedge \Lambda, \qquad [\Lambda,E]=0
\end{equation}
where the brackets are the Schouten-Nijenhuis brackets. Then, the bracket 
\begin{equation} \label{JB}
\{f,g\}=\Lambda(df,dg)+fE(g)-gE(f)
\end{equation} 
turns out to be a Jacobi bracket. The inverse of this assertion is also true; that is, if a Jacobi bracket is defined through the pair $\Lambda$ and $E$ as in \eqref{JB}, then both conditions given in \eqref{JB-CC} must be satisfied in order to preserve the Jacobi identity. Notice that a Jacobi bracket does not necessarily satisfy the Leibniz rule. If it satisfies the Leibniz identity, it turns into a Poisson bracket \cite{vaisman-book}. A good example of a Jacobi manifold which is not a Poisson manifold is precisely a l.c.s. manifold. Let us exhibit this in more detail.
 
Let us consider now a l.c.s. manifold $(M,\omega)$ with a Lee form $\theta$. Whether referring to the almost symplectic two-form $\omega$, or to the local symplectic two-forms $\omega_\alpha$ defined on the local charts, it is possible to show that $M$ is a Jacobi manifold. Let us start with the global definition of l.c.s. manifolds and, referring to the musical mapping (\ref{mus-iso}), define a bivector $\Lambda$ on $M$.
For two elements $\mu$ and $\nu$ of $\Omega^1(M)$, we construct the bivector as follows:
\begin{equation} \label{gamma}
\Lambda(\mu,\nu)=\omega(\omega^\sharp(\mu), \omega^\sharp(\nu)).
\end{equation} 
By direct calculation, one can show that the pair $(\Lambda,Z_\theta)$ satisfies the conditions in (\ref{JB-CC}) assuming $\Lambda$ as in (\ref{gamma}), and the vector field $E$ is assumed to be the Lee vector field $Z_\theta$ in (\ref{Lee-v-f}), see \cite{LeSa17}.  Accordingly, Jacobi bracket of two functions is determined through
\begin{equation}
\{f,g\}=\Lambda(df,dg)+fZ_\theta(g)-gZ_\theta(f)=\omega(X_f,X_g).
\end{equation} 
Notice that, $X_f$ and $X_g$ are the Hamiltonian vector fields in the form of (\ref{vflcs}). 
 
Let us now describe the local picture. Consider the local symplectic manifold $(U_\alpha,\omega_\alpha)$ and two local functions $f_\alpha$ and $g_\alpha$. Then, we define the local bracket
\begin{equation} \label{locJac}
 \{f\vert_\alpha,g\vert_\alpha\}=e^{-\sigma_{\alpha}} \{e^{\sigma_\alpha} f_{\alpha},e^{\sigma_\alpha}g_{\alpha}\},
\end{equation} 
where the bracket on the right hand side is the canonical Poisson bracket defined by means of the local symplectic two-form $\omega_\alpha$. The functions inside the bracket (\ref{locJac}) are the local realizations $f\vert_\alpha=e^{\sigma_{\alpha}}f_\alpha$ and $g\vert_\alpha=e^{\sigma_{\alpha}}g_\alpha$ of two global function $f$ and $g$, respectively. 
 
The following lemma will show that the notion of Lagrangian submanifold on a l.c.s. manifold concides with the notion of Lagrangian submanifolds of a Jacobi manifold \cite{LiMa-book,vaisman}.

\begin{lemma}
 It is fulfilled
 \begin{equation}
     TL^{\bot}=\sharp_{\Lambda}(TL^{\circ}),
 \end{equation}
 where $TL^{\bot}$ is given in \eqref{lagsub} and $\sharp_{\Lambda}$ is the musical mapping defined to be $\sharp_{\Lambda}(\mu)(\nu)=\Lambda(\mu,\nu)$.
 \end{lemma}
  \begin{proof}
It directly follows from the identity $\sharp_\Lambda = - \omega^\sharp
$. To see this more explicitly, let $\mu$ be an element of $TL^{\circ}$ and consider $\sharp_{\Lambda}(\mu)$ in $\sharp_{\Lambda}(TL^{\circ})$. For an arbitrary element $w$ in $TL$, we perform the following calculation
\begin{equation}
\omega(\sharp_{\Lambda}(\mu),w)
=
-\omega(\omega^\sharp(\mu),w)
= - \langle\omega^\flat(\omega^\sharp(\mu)),w\rangle=
-\langle \mu, w \rangle = 0
\end{equation}
showing that $\sharp_{\Lambda}(\mu)$ is an element of $TL^{\bot}$. 
Reciprocally, let $v$ be an element in $TL^{\bot}$. Non-degeneracy of $\sharp_{\Lambda}$ enables us to determine a unique $\mu$ in the cotangent bundle so that $v=\sharp_{\Lambda}(\mu)$. Alternatively, we can write this correspondence as $\mu=-\omega^\flat(v)$. Let us prove that $\mu$ is an element of $TL^\circ$. See that, for $w$ in $TL$, we have
\begin{equation}
\langle \mu, w \rangle= -\langle \omega^\flat(v), w \rangle=-\omega(v,w)=0.
\end{equation}
\end{proof}
\section{Hamilton-Jacobi theory on l.c.s. manifolds}\label{HJlcs2}

Here, we will present a Hamilton-Jacobi theory for the case of l.c.s. cotangent bundles, since we are interested in its applications in mechanics. First, we will start by proposing a Hamilton-Jacobi theory locally, that is, in the defined subsets $U_{\alpha}$. After, we will
present a global approach in terms of the Lichnerowicz-deRham differential. Then we close this section by a comparison of the local and global theories. 

\subsection{The local picture of a l.c.s. HJ theory}

Our interest focuses on the l.c.s. cotangent bundles preented in Subsection \ref{Sec-lcs-cot}. Accordingly, we start with the l.c.s. manifold $T^*_\theta Q$ given in (\ref{T*_Q}). Consider an open covering $\{V_\alpha\}$ on the base manifold $Q$ so that, on each chart $V_\alpha$, the closed one-form $\vartheta$ can be written as $\vartheta = d\mu_\alpha$ for a real valued function $\mu_\alpha$. Pull each open set $V_\alpha$ back to the cotangent bundle $T^*Q$ by means of the canonical projection $\pi_Q$ in order to arrive at an open covering $\{U_\alpha\}:=\{\pi_Q^{*}(V_\alpha)\}$ of $T^*Q$. In each chart $U_\alpha$, function $\sigma_\alpha$ determining the conformal character of the almost symplectic structure turns out to be $\sigma_\alpha=\mu_\alpha\circ \pi_Q$. The natural bundle coordinates $(q^i_{(\alpha)},p_i^{(\alpha)})$ on $U_\alpha$ are the Darboux' coordinates for the symplectic form $\omega_\alpha$, that is 
\begin{equation}
\omega_\alpha=dq^i_{(\alpha)}\wedge dp_i^{(\alpha)}. 
\end{equation}
We denote the restriction of the cotangent bundle projection $\pi_Q$ to a chart $U_\alpha$ by $\pi_\alpha$. Then, the fibration $\pi_\alpha$ turns out to be the projection to the first factor that is
\begin{equation}
\pi_\alpha: U_\alpha  \mapsto V_\alpha:(q^i_{(\alpha)},p_i^{(\alpha)})\longrightarrow (q^i_{(\alpha)}). 
\end{equation}
So that the restriction of a section $\gamma$ of $\pi_Q$ is given by
\begin{equation}
\gamma_\alpha: V_\alpha \mapsto U_\alpha: (q^i_{(\alpha)}) \longrightarrow (q^i_{(\alpha)},(\gamma_\alpha)_i),
\end{equation}
where $(\gamma_\alpha)_i$'s are real valued functions on $V_\alpha$. 
 
Consider a Hamiltonian vector field $X_h$ on $T^*_\theta Q$ determined through the Hamilton's equation (\ref{semiglobal3}). On each local chart $U_\alpha$, we have local vector fields $X_\alpha$ satisfying the local identity (\ref{geohamalpha}) for Hamiltonian functions $h_\alpha$. 
 Accordingly, we define a vector field on $V_\alpha$ as follows
\begin{equation} \label{X_H-gamma-loc}
 X_{\alpha}^{\gamma_{\alpha}}=T\pi_\alpha\circ X_{\alpha}\circ\gamma_{\alpha},
\end{equation} 
where $T\pi_\alpha$ is the tangent mapping from $TU_\alpha$ to $TV_\alpha$. 
The following diagram summarizes this discussion. 
\begin{equation} \label{diag-1}
\xymatrix{ U_\alpha 
\ar[dd]^{\pi_{\alpha}} \ar[rrr]^{X_{{\alpha}}}&   & &TU_{\alpha}\ar[dd]^{T\pi_{\alpha}}\\
  &  & &\\
V_\alpha \ar@/^2pc/[uu]^{\gamma_{\alpha}}\ar[rrr]^{X_{{\alpha}}^{\gamma_{\alpha}}}&  & & TV_\alpha  }
\end{equation} 
Let us state the Hamilton-Jacobi theorem valid for this local picture. 
This Hamilton-Jacobi theorem is a particular case of more general theorem that was proved e.g. in \cite{MDV}. 
\begin{theorem}\label{HJ1}
Consider the local symplectic structure $(U_\alpha ,\omega_\alpha)$. Let $\gamma_{\alpha}: V_\alpha \to U_\alpha$ be a section, whose image space $\operatorname{Im}\gamma_{\alpha}$ is a Lagrangian submanifold of $U_\alpha$. Then the following  conditions are equivalent:
\begin{enumerate}
\item The vector fields $X_{{\alpha}}$ and $X_{{\alpha}}^{\gamma_{\alpha}}$ are $\gamma_{\alpha}$-related,
that is
\begin{equation} \label{hjlocal-}
T\gamma_{\alpha} \circ X_{{\alpha}}^{\gamma_{\alpha}}=X_\alpha\circ\gamma_{\alpha}.
\end{equation}
\item  
\begin{equation}\label{hjlocal}
d (h_\alpha\circ\gamma_\alpha)=0.
\end{equation}
\end{enumerate}
\end{theorem}
\subsection{A global picture for a l.c.s. Hamilton-Jacobi theory}

Now, we want to arrive at the global picture. For this, recall the global definition of l.c.s. manifold $T^*_\theta Q$. Consider a Hamiltonian vector field $X_h$ defined through the equation (\ref{semiglobal3}). Let us consider now a section $\gamma$ of the cotangent bundle and define a vector field on $Q$ as
\begin{equation} \label{X_H-gamma}
 X_{H}^{\gamma}=T\pi\circ X_{H}\circ\gamma.
\end{equation}
One can define the vector field $X_{H}^{\gamma}$ by the commutation of the following diagram
\[
\xymatrix{ T_{\theta}^*Q
\ar[dd]^{\pi_Q} \ar[rrr]^{X_{h}}&   & &TT_{\theta}^*Q  \ar[dd]^{T\pi_Q}\\
  &  & &\\
Q \ar@/^2pc/[uu]^{\gamma}\ar[rrr]^{X_{h}^{\gamma}}&  & & TQ}
\]
We are ready to write the Hamilton-Jacobi theorem for l.c.s. cotangent bundles. 
\begin{theorem}
\label{HJT} Consider a one-form $\gamma$ whose image is a Lagrangian submanifold of the locally conformally symplectic manifold $T^*_\theta Q$ with respect to the almost symplectic two-form $\Omega_\theta$, that is $d_\vartheta \gamma=0$. Then, the following  conditions are equivalent:
\begin{enumerate}
\item The vector fields $X_{h}$ and $X_{h}^{\gamma}$ are $\gamma$-related,
that is
\begin{equation}
T\gamma \circ X_{h}^{\gamma}=X_h\circ\gamma.
\end{equation}
\item  
\[
d_\vartheta (h\circ\gamma)=0.
\]
\end{enumerate}
\end{theorem}
\begin{proof} $(1)\Rightarrow (2)$: Recall the characterization of Lagrangian submanifolds on $T^{*}_{\theta}Q$ presented in \eqref{LagSubT*Q}. It says that the image of a one-form section $\gamma$ on $Q$ is Lagrangian submanifold if and only if $d_{\vartheta}\gamma=0$. 
Let us first try to examine the second condition:
\begin{equation} \label{gamma-pull}
\begin{split}
d_\vartheta(h\circ \gamma)&=d(h\circ \gamma)-(h\circ \gamma)\vartheta
=d(\gamma^*h)-\gamma^*h\vartheta
\\
&=\gamma^*dh-\gamma^*h\gamma^*\theta=\gamma^*(dh-h\theta)
=\gamma^*(d_\theta h).
\end{split}
\end{equation}
So that, we can rewrite (2) as $\gamma^*(d_\theta h)=0$. 
Assume now that the first condition holds, that is $X_{h}$ and $X_{h}^{\gamma}$ are $\gamma$-related, or equivalently $X_h=\gamma_*X_h^\gamma$ over the image space of $\gamma$. Then, we have 
\begin{equation}
\gamma^*(d_\theta h)=\gamma^*(\iota_{X_h}\Omega_\alpha)=\iota_{X_h^\gamma}\gamma^*\Omega_\alpha=0,
\end{equation}
where we have substituted the definition in (\ref{semiglobal3}) in the first equality.
 
\noindent $(2)\Rightarrow (1)$:
Conversely, assume that the second condition holds. To prove that the first condition, we define a vector field $D=X_h-\gamma_*X_h^\gamma$. If the vector field $D$ is identically zero, we have proven the theorem. To show this, we first need to see that $D$ is a vertical vector field over the image space of $\gamma$ with respect to the cotangent bundle projection $\pi_Q$, that is
\begin{eqnarray*}
T\pi_Q\circ D\circ \gamma &=&T\pi_Q (X_h-\gamma_*X_h^\gamma)\circ \gamma
\\&=&T\pi_Q  \circ X_h \circ \gamma - T\pi_Q \circ T\gamma \circ X_h^\Lambda
\\
&=&T\pi_Q  \circ X_h \circ \gamma -T\pi_Q \circ T\gamma \circ T\pi_Q \circ 
X_h \circ \gamma
\\
&=&T\pi_Q  \circ X_h \circ \gamma -T\pi_Q \circ 
X_h \circ \gamma=0.
\end{eqnarray*}
Now see that $D=X_h-\gamma_*X_h^\gamma$ vanishes identically on the image space of $\gamma$ for the vector fields in the form $\gamma_*Y$ for any vector field $Y$ on $Q$. 
\begin{eqnarray*}
\Omega_\alpha(X_h-\gamma_*X_h^\gamma,\gamma_*Y)&=&\Omega_\alpha(X_h,\gamma_*Y)-\Omega_\alpha(\gamma_*X_h^\gamma,\gamma_*Y)
\\
&=& \iota_{X_h}\Omega_\alpha(\gamma_*Y)-\gamma^*(\Omega_\alpha)(X_h^\gamma,Y)
\\
&=& \gamma^*(\iota_{X_h}\Omega_\alpha)(Y)-\gamma^*(\Omega_\alpha)(X_h^\gamma,Y)=0
\end{eqnarray*}
where we have employed the second condition for the first term in the last line, whereas we employed the Lagrangian submanifold property of $\gamma$  in the second term of the last line.
Notice also that we have the following decomposition of the vector spaces
\begin{equation}
T_{\gamma(q)} T^{*} Q=T_{\gamma(q)}(\operatorname{Im} \gamma) \oplus \mathrm{V}_{\gamma(q)} \pi_{Q}.
\end{equation}
$\mathrm{V}_{\gamma(q)} \pi_{Q}$ is a Lagrangian subspace so that the proof now follows the non-degeneracy of $\Omega_\theta$.
\end{proof}
\subsection{Comparison of local and global pictures} \label{comparison}

We have exhibited two HJ theorems to describe Hamiltonian dynamics on l.c.s. manifolds. One theorem is concerned with the local description of dynamics, whilst the other provides a global picture of such dynamics. Now, we would like to establish the link between these theories, namely, we need to show that restricting to the local charts of the HJ Theorem \ref{HJT}, one can retrieve the HJ Theorem \ref{HJ1}. For this, consider the l.c.s. manifold $(T^*_\theta Q,\Omega_\theta)$ with Lee-form $\theta$. Notice that, while referring to the l.c.s. structures on the cotangent bundles, we denote the l.c.s. two-form by $\Omega_\theta$ instead of $\omega$. On the other hand, we will still make use of the same notation in local coordinates, i.e., on a local chart $U_\alpha$, $\omega\vert_\alpha$ we denote the restriction of the almost symplectic two-form by $\Omega_\theta$.  

If the image spaces of $\gamma_\alpha$'s are Lagrangian submanifolds of the symplectic manifolds $(U_\alpha,\omega_\alpha)$ then $\gamma_\alpha^*\omega_\alpha=0$. By substitution of the local realization of the almost symplectic form $\omega\vert_\alpha=e^{\sigma_\alpha}\omega_\alpha$ into this identity, we see that the image space of $\gamma_\alpha$ is also a Lagrangian submanifold of the almost symplectic manifold $(U_\alpha,\omega\vert_\alpha)$. That is $\gamma_\alpha^*(\omega\vert_\alpha)=0$. By gluing up the local sections $\gamma_\alpha$, we arrive at a one-form $\gamma$ on $Q$,  whose image is a Lagrangian submanifold of $T_\theta^*Q$. The Poincar\'e's lemma states that $\gamma$ must be closed. Inversely, if the image of $\gamma$ is a Lagrangian submanifold of $T_\theta^*Q$, then, in each local chart, the image space of its restriction $\gamma_\alpha$ is a Lagrangian submanifold of $(U_\alpha,\omega\vert_\alpha)$. So, it is also a Lagrangian submanifold of $(U_\alpha,\omega_\alpha)$, since they differ only by a non-vanishing positive function. Therefore, we have proven that the assumptions in Theorems     (\ref{HJ1}) and (\ref{HJT}) are equivalent. Let us give further details.

Start with the first condition. On a chart $U_\alpha$, the local picture of $X_h$, $\gamma$ and $\pi_Q$ are given by $X_h|_\alpha=X_\alpha$, $\gamma|_\alpha=\gamma_\alpha$, and $\pi_Q\vert_\alpha=\pi_\alpha$, respectively. This shows that the local realization of $X_h^\gamma$ in (\ref{X_H-gamma}) given by $X_h^\gamma\vert_\alpha$ equals the local vector field $X^{\gamma_\alpha}_\alpha$ in (\ref{X_H-gamma-loc}), that is 
\begin{equation}
X_{h}^{\gamma}\vert_\alpha=(T\pi_Q\circ X_{h}\circ\gamma)\vert_\alpha= T\pi_\alpha\circ X_{\alpha}\circ\gamma_\alpha =X_\alpha^{\gamma_\alpha}.
\end{equation}   
It is easy to see now that the first conditions in both theorems are equivalent. For the second condition in both theorems, recall the identity $d_\vartheta(h\circ \gamma)=\gamma^*(d_\theta h)$ given in (\ref{gamma-pull}). Accordingly, we compute 
\begin{eqnarray*}
0=d_\vartheta (h\circ \gamma)&=&\gamma^*(d_\theta h)=\gamma^*\Big(d h  -  h\theta \Big) = \gamma_\alpha^*(d(e^{\sigma_\alpha} h_{\alpha})  -  (e^{\sigma_\alpha}h_\alpha)d\sigma_\alpha )  
\\ &=& {\gamma_\alpha}^*\Big(d(e^{\sigma_\alpha}h_\alpha)  -  (e^{\sigma_\alpha}h_\alpha)d\sigma_\alpha \Big) \\ &=& {\gamma_\alpha}^*\Big( h_\alpha e^{\sigma_\alpha}d\sigma_\alpha +e^{\sigma_\alpha}d h_\alpha   -  (e^{\sigma_\alpha}h_\alpha)d\sigma_\alpha \Big) \\ &=& {\gamma_\alpha}^*\Big( e^{\sigma_\alpha}d h_\alpha  \Big)= (e^{\sigma_\alpha\circ\gamma_\alpha}  ) d( h_\alpha\circ\gamma_\alpha)   
\end{eqnarray*}
that proves the equivalency by a local conformal factor.

\section{An example} \label{Exp}

Consider two dimensional punctured Euclidean space $Q=\mathbf{R}^2-\{0\}$ where $0$ is the origin. In order to define a l.c.s. manifold structure on the cotangent bundle as described in Subsection \ref{Sec-lcs-cot}, we first introduce the following closed one-form 
\begin{equation} \label{Lee-ex}
\vartheta=2\frac{xdy-ydx}{x^2+y^2}
\end{equation}
which fails to be exact on the whole $Q$. Introduce the Darboux' coordinates $(x,y,p_x,p_y)$ on the cotangent bundle $T^*Q$ and, by referring to the definition (\ref{omega_theta}), consider the following non-degenerate two-form 
\begin{equation} \label{Ex-Omega}
\Omega_\theta=\Omega_Q+\theta\wedge \theta_Q=dx\wedge dp_x+dy\wedge dp_y - 2\frac{yp_y+xp_x}{x^2+y^2}dx\wedge dy.
\end{equation}
Here, the one-form $\theta$ is the pull-back of the one-form $\vartheta$ in \eqref{Lee-ex} by means of the cotangent bundle projection $\pi_Q$ and, in the present setting, it looks like the same with $\vartheta$.
After a direct calculation it is immediate to observe that $\Omega_\theta$ is a locally conformal symplectic two-form with the Lee form $\theta$. Later, we introduce the following quadratic Hamiltonian function 
\begin{equation} \label{Ex-ham-func}
h=\frac{1}{2}(p_x^2+p_y^2)
\end{equation}
on the cotangent bundle $T^*Q$. Recall the Hamilton's equation \eqref{semiglobal3} in the l.c.s. framework defined in terms of the Lichnerowicz-deRham differential $d_\vartheta$.  
For the present case, the Hamilton's equation takes the following particular form
\begin{equation}
\dot{x}=p_x,\quad \dot{y}=p_y,\quad \dot{p}_x=\frac{-yp^2_x-2yp^2_y-xp_xp_y}{x^2+y^2}, \quad \dot{p}_y=\frac{-xp^2_y-2xp^2_x -yp_xp_y}{x^2+y^2}.
\end{equation}

We write this system in a local coordinate chart. For this end, choose the polar coordinates $x=r\cos\phi$, $x=r\sin\phi$ on an open chart $V_\alpha$ in $Q$. In these coordinates, the one-form $\vartheta$ in \eqref{Lee-ex} turns out to be an exact one-form $2d\phi$. We consider the induced local coordinates $(r,\phi,p_r,p_\phi)$ on $U_\alpha$. In this realization, the Lee form is computed to be $\theta=2d\phi$. This gives that the local function, determining the conformal character of the system, is $\sigma_\alpha=2\phi$. The two-form exhibited in (\ref{Ex-Omega}) reduces to
\begin{equation}
\Omega_\theta\vert_\alpha=dr\wedge dp_r+d\phi\wedge dp_\phi-2p_rdr\wedge d\phi.
\end{equation}
This two-form is not closed but the following one, which is defined according to the formula in (\ref{ua}),
 \begin{equation}
\Omega_\alpha=e^{-\sigma_\alpha}\Omega_\theta\vert_\alpha=e^{-2\phi} \Omega_\theta\vert_\alpha=e^{-2\phi}\Big(dr\wedge dp_r+d\phi\wedge dp_\phi-2p_rdr\wedge d\phi\Big)
\end{equation}
is closed. Notice that, for the symplectic two-form $\Omega_\alpha$, the canonical coordinates can be determined by  $\bar{p}_r=e^{-2\phi}p_r$, $\bar{p}_\phi=e^{-2\phi}p_\phi$ where the base coordinates remain the same. Nevertheless, we insist to use non-canonical polar coordinates on symplectic pair $(U_\alpha,\Omega_\alpha)$. In this coordinate system, the Hamiltonian function (\ref{Ex-ham-func}) is written as 
\begin{equation}
h\vert_\alpha=\frac{1}{2}\left( p_{r}^2+\frac{1}{r^2}p^2_\phi \right)
\end{equation}
and then, according to (\ref{glueHamFunc}), we define a local function 
\begin{equation}
h_\alpha=e^{-\sigma_\alpha}h\vert_\alpha=\frac{1}{2}e^{-2\phi}\left( p_{r}^2+\frac{1}{r^2}p^2_\phi \right)
\end{equation}
on the local symplectic manifold $U_\alpha$. The Hamiltonian dynamics generated by the Hamiltonian function $h_\alpha$, according to the local Hamilton's equation \eqref{geohamalpha},
is computed to be
\begin{equation}
\dot{r}=p_r, \quad \dot{\phi}=\frac{1}{r^2}p_\phi, \quad \dot{p_r}=\frac{1}{r^2}p_\phi(2 p_r+ \frac{1}{r}p_\phi), \quad \dot{p_\phi}=-p_{r}^2+\frac{1}{r^2}p^2_\phi.
\end{equation}

Let us apply now the local HJ theorem \ref{HJ1} for the present setting. To this end, we start with a section 
\begin{equation}
\gamma_\alpha=\xi(r,\phi)dr+\eta(r,\phi)d\phi
\end{equation}
of the local bundle. We ask that the image space of $\gamma_\alpha$ be a Lagrangian submanifold of the local symplectic manifold $(U_\alpha,\Omega_\alpha)$ that is $\gamma^*\Omega_\alpha$ vanishes identically. This determines the following condition on the coefficient functions  
\begin{equation}\label{lagsub2}
\frac{\partial\xi}{\partial\phi}- \frac{\partial\eta}{\partial r}-2\xi =0
\end{equation}
of $\gamma$. The second condition $\d (h_\alpha\circ \gamma_\alpha)=0$ in the local HJ theorem \ref{HJ1} gives the following HJ equations
\begin{equation} \label{Ex-locHJ}
\begin{split}
 \xi\frac{\partial\xi}{\partial\phi}+\frac{1}{r^2} \eta\frac{\partial\eta}{\partial\phi}&=\xi^2+\frac{1}{r^2}\eta^2, \\  \xi\frac{\partial\xi}{\partial r}+ \frac{1}{r^2}\eta\frac{\partial\eta}{\partial r} &=\frac{1}{r^3}\eta^2.
\end{split}
\end{equation}
As stated in Theorem \ref{HJ1}, we can alternatively arrive at this system of equations by referring the identity presented in (\ref{hjlocal-}). For this, we first compute the right hand side of the equality (\ref{hjlocal-}) that is
\begin{equation}
X_\alpha^{\gamma_\alpha}=  \xi\partial_{r}+\eta\frac{1}{r^2}\partial_{\phi}   
\end{equation}
and then the left hand side of the identity
\begin{equation}
T{\gamma_\alpha} X_\alpha^{\gamma_\alpha}=   \xi\partial_{r}+\eta\frac{1}{r^2}\partial_{\phi} + \Big( \frac{\partial\xi}{\partial r}\xi+ \frac{\partial\xi}{\partial \phi}\eta \frac{1}{r^2} \Big)\partial_{p_r}  + \Big(   \frac{\partial\eta}{\partial r}\xi  +  \frac{\partial\eta}{\partial \phi}\eta \frac{1}{r^2} \Big)\partial_{p_\phi}.
\end{equation}
Now, it is a matter of a direct calculation to arrive at the system (\ref{Ex-locHJ}).

Let us now apply the global HJ Theorem \ref{HJT} to the present example. Start with a one-form 
\begin{equation}
\gamma=\beta(x,y)dx+\rho(x,y)dy
\end{equation}
satisfying $d_\vartheta \gamma=0$ in order to guarantee that the image space of $\gamma$ is a Lagrangian submanifold of the l.c.s. manifold $T^*_\theta Q$ with the two-form $\Omega_\theta$ in (\ref{Ex-Omega}). That is we have 
\begin{equation}
\frac{\partial \rho}{\partial x} - \frac{\partial \beta}{\partial y} + \frac{2x\beta-2y\rho}{x^2+y^2}=0.
 \end{equation}
Referring to the second condition in Theorem \ref{HJT}, we write HJ equation as 
\begin{equation}
d_\vartheta(h\circ \gamma)=\frac{1}{2}d( \beta^2+\rho^2)-(\beta^2+\rho^2)\left( \frac{xdy-ydx}{x^2+y^2}\right)=0.
\end{equation}
Explicitly, we compute the following system of equations
\begin{equation}
\begin{split}
\beta\frac{\partial \beta}{\partial x} + \rho\frac{\partial \rho}{\partial x} + ( \beta^2+\rho^2)\frac{y}{x^2+y^2}&=0,
\\
\beta\frac{\partial \beta}{\partial y} + \rho\frac{\partial \rho}{\partial y} - ( \beta^2+\rho^2)\frac{x}{x^2+y^2}&=0.
\end{split}
\end{equation}

\section{Complete solutions} \label{compsolu}
A complete solution for a HJ theory for a Hamiltonian system on a l.c.s. manifold $(T_{\theta}^{*}Q,\Omega_{\theta})$ is a diffeomorphism 
$$\Phi : Q \times \mathbb{R}^n \longrightarrow T^*_ \theta Q$$
  such that for a set of parameters $\lambda=(\lambda_1,\dots,\lambda_n)$ in $\mathbb{R}^n$, the section 
  $\Phi_\lambda$ from $Q$ to $T^*_\theta Q$ is a solution for the Hamilton-Jacobi problem for a given Hamiltonian function $h$. That is $\Phi_{\lambda}$ satisfies these two conditions:
\begin{enumerate}
    \item $d_{\theta}\Phi_{\lambda}=0$, i.e, the image of $\Phi_{\lambda}$ is Lagrangian submanifold with respect to $\Omega_\theta$.
    \item $d_\vartheta(h\circ \Phi_{\lambda})=0.
    $
\end{enumerate}
\bigskip

\noindent Consider the following commutative diagram
\begin{equation}
\xymatrix{ Q\times \mathbb{R}^n
\ar[d]^{\nu} \ar[rr]^{\Phi}&    &T^{*}_{\theta} Q\ar[d]^{f_i}
\\ \mathbb{R}^n\ar[rr]^{\nu_i} && \mathbb{R} }
\end{equation}
\noindent
where $\nu$ is the canonical projection on the second factor, $\nu_i$ is the projection into the $i$th-factor and $f_i$ is the function defined by 
\begin{equation}
f_i =\nu_i\circ\nu\circ \Phi^{-1}. 
\end{equation}
\begin{theorem}
If $\Phi$ is a complete solution of the Hamilton-Jacobi problem on $T_\theta^*Q,$ then, the functions defined as $f_i$ commute with respect to the Jacobi bracket, say,
\begin{equation}
    \{f_i,f_j\}=0,\qquad \forall i,j.
\end{equation}
\end{theorem}
\begin{proof}
 It is clear that
 \begin{equation}
  \operatorname{Im}\Phi_{\lambda}=\bigcap_{i=1}^n f_i^{-1}(\lambda_i).
 \end{equation}
  Therefore, if an element $z$ is in $\operatorname{Im}\Phi_{\lambda}$  then   $z=\Phi_{\lambda}(q)$  for a point $q\in Q$ and  we have
 \begin{equation*}
  f_i(\Phi_{\lambda}(q))=f_i(\Phi(q,\lambda))=\lambda_i.
 \end{equation*}
 This means that $f_i$ is constant on $\operatorname{Im}\Phi_{\lambda}$ and therefore $df_i$ vanishes on $T(\operatorname{Im}\Phi_{\lambda})$. 
 But $\Phi_{\lambda}$ is a solution of the Hamilton--Jacobi equation, then  $\operatorname{Im}\Phi_{\lambda}$ is a Lagrangian submanifold and we have that
 \begin{equation*}
  {\sharp}_{\Lambda} T(\operatorname{Im}\Phi_{\lambda})^{\circ}=T(\operatorname{Im}\Phi_{\lambda})
 \end{equation*}
 that implies ${\sharp}_{\Lambda}(df_i)(f_j)=0$ for all $i,j$. But
 ${\sharp}_{\Lambda}(df_i)=X_{f_i}-f_iZ_{\theta}$, so that
 \begin{equation}
 0={\sharp}_{\Lambda}(df_i)(f_j)=X_{f_i}(f_j)-f_iZ_\theta(f_j)=\{f_i,f_j\}.
 \end{equation}
\end{proof}

\section{Discussion, and Future Works}

In this work, we have presented Hamilton-Jacobi theory for the Hamiltonian dynamics on l.c.s. manifolds in both the local and the global frameworks in Theorems \ref{HJ1} and \ref{HJT}, respectively.   We have, theoretically, exhibited the passages between these two realizations in Subsection \ref{comparison}. We have discussed all these geometries on a concrete example in Section \ref{Exp}. Additionally, we have presented complete solutions of the Hamilton-Jacobi problem in Section \ref{compsolu}. We plan to pursue investigations on Hamilton-Jacobi theories for locally conformal settings in the following headlines:
\begin{itemize} 
\item $k$-symplectic manifolds are generalizations of 
symplectic manifolds that are appropriate for field 
theories \cite{Awane}. There is Hamilton-Jacobi theory 
for the dynamics on $k$-symplectic manifolds 
\cite{LDMSV,LSV}. 
It looks interesting to investigate the locally 
conformal setting for the case of $k$-symplectic 
geometry. We also wish to study the Hamiltonian dynamics and the Hamilton-Jacobi theory on this framework.

\item Reduction of the Hamiltonian dynamics 
on l.c.s. manifolds under a group of symmetries
has been studied in the literature \cite{HaRy01,IbLeMa97,St19} whereas reduction of the Hamilton-Jacobi formalism under symmetry has been exhibited in \cite{LDV}. It could be interesting to merge to these two reduction procedures to explore possible reductions of the Hamilton-Jacobi formalisms for l.c.s. geometry given in this paper under Lie group actions.
 
\end{itemize}

\section{Acknowledgments}
This work has been partially supported by MINECO grants MTM2016-76-072-P and the ICMAT Severo Ochoa Project SEV-2011-0087 and SEV-2015-0554.

\end{document}